\documentclass[11pt]{article}
\usepackage{graphicx}
\usepackage{multirow}

\usepackage[T1]{fontenc}
\usepackage[sc]{mathpazo}
\usepackage{amsmath}
\usepackage{amssymb}
\usepackage{enumerate}
\usepackage{amsthm}
\usepackage{amsfonts,mathrsfs}

\usepackage{hyperref}
\hypersetup{pdfpagemode=UseNone}

\newcommand{\tinyspace}{\mspace{1mu}}

\newcommand{\abs}[1]{\left\lvert\tinyspace #1 \tinyspace\right\rvert}

\newcommand{\norm}[1]{\left\lVert\tinyspace #1 \tinyspace\right\rVert}

\newcommand{\setft}[1]{\mathrm{#1}}

\newcommand{\density}[1]{\setft{D}\left(#1\right)}

\newcommand{\supp}{{\operatorname{supp}}}

\def\I{\mathbb{1}}

\newenvironment{mylist}[1]{\begin{list}{}{
    \setlength{\leftmargin}{#1}
    \setlength{\rightmargin}{0mm}
    \setlength{\labelsep}{2mm}
    \setlength{\labelwidth}{8mm}
    \setlength{\itemsep}{0mm}}}
    {\end{list}}

\newcommand{\Pa}[1]{\left(#1\right)}

\newcommand{\Br}[1]{\left[#1\right]}
\newcommand{\set}[1]{\{#1\}}
\newcommand{\Set}[1]{\left\{#1\right\}}

\DeclareMathOperator{\trace}{Tr}

\newcommand{\Ptr}[2]{\trace_{#1}\Pa{#2}}

\newcommand{\Tr}[1]{\Ptr{}{#1}}

\def\cH{\mathcal{H}}\def\cI{\mathcal{I}}
\def\cM{\mathcal{M}}

\def\rF{\mathrm{F}}

\def\rS{\mathrm{S}}

\newtheorem{thrm}{Theorem}[section]
\newtheorem{lem}[thrm]{Lemma}
\newtheorem{prop}[thrm]{Proposition}

\theoremstyle{definition}

\numberwithin{equation}{section}

\newcounter{questionnumber}

\begin{document}

\title{A lower bound on the fidelity
between two states in terms of their trace-distance and max-relative
entropy}

\author{Lin Zhang\footnote{E-mail: godyalin@163.com; linyz@zju.edu.cn}\\
  {\it\small Institute of Mathematics, Hangzhou Dianzi University, Hangzhou 310018, PR~China}\\
  Kaifeng Bu, Junde Wu\\
  {\it\small Department of Mathematics, Zhejiang University, Hangzhou 310027, PR~China}}
\date{}
\maketitle
\maketitle \mbox{}\hrule\mbox\\
\begin{abstract}

Fidelity is a fundamental and ubiquitous concept in quantum
information theory. Fuchs-van de Graaf's inequalities deal with
bounding fidelity from above and below. In this paper, we give a
lower bound on the quantum fidelity between two states in terms of
their trace-distance and their max-relative entropy.
\end{abstract}
\maketitle \mbox{}\hrule\mbox

\section{Introduction}

The \emph{fidelity} between two quantum states, represented by
density operators $\rho$ and $\sigma$, is defined as
\begin{eqnarray*}
\rF(\rho, \sigma ) := \Tr{\sqrt{\sqrt{\rho}\sigma\sqrt{\rho}}}.
\end{eqnarray*}
Note that both density operators here are taken from
$\density{\cH_d}$, the set of all positive semi-definite operator
with unit trace on a $d$-dimensional Hilbert space $\cH_d$. The
squared fidelity above has been called \emph{transition probability}
by Uhlmann \cite{Uhlmann76,Uhlmann2011}, operationally it is the
maximal success probability of changing a state to another one by a
measurement in a larger quantum system. The fidelity is also
employed in a number of problems such as quantifying entanglement
\cite{Vedral}, and quantum error correction \cite{Kosut}, etc.

For quantum fidelity, the well-known Fuchs-van de Graaf's inequality
states that: For arbitrary two density operators $\rho$ and $\sigma$
in $\density{\cH_d}$,
\begin{eqnarray}\label{eq:Graaf}
1-\frac12\norm{\rho-\sigma}_1 \leqslant \rF(\rho, \sigma) \leqslant
\sqrt{1-\frac14\norm{\rho-\sigma}^2_1},
\end{eqnarray}
which established a close relationship between the trace-norm of the
difference for two density operators and their fidelity
\cite{Watrous08}, where $\norm{\rho-\sigma}_1:=
\Tr{\sqrt{(\rho-\sigma)^2}}$.

The Fuchs-van de Graaf's inequality can not be improved because it
is tight. For any value of $\norm{\rho-\sigma}_1$ there exists a
pair of states saturating the inequality. However, by supplying
additional information about the pair it is possible to obtain a
tighter lower bound on the fidelity. In this paper we consider
supplying the max-relative entropy between the states as additional
information. The max-relative entropy is defined as
$$
\rS_{\max}(\rho||\sigma) :=\min\set{\gamma: \rho\leqslant
e^\gamma\sigma}.
$$
Clearly by the definition, $\rS_{\max}(\rho||\sigma)<+\infty$ if and
only if the support of $\rho$ is contained in that of $\sigma$, and
$e^{\rS_{\max}(\rho||\sigma)} =
\lambda_{\max}(\sigma^{-1/2}\rho\sigma^{-1/2})$, where
$\lambda_{\max}(X)$ means the maximal eigenvalue of the operator
$X$.

Throughout this paper, we denote the classical fidelity between two
probability distributions $p=\set{p_j}^n_{j=1}$ and
$q=\set{q_j}^n_{j=1}$ is $\rF(p,q):=\sum^n_{j=1}\sqrt{p_jq_j}$.
Also, the $\ell_1$-norm between $p$ and $q$ is defined by
$\norm{p-q}_1:=\sum^n_{j=1}\abs{p_j-q_j}$. We also use the notion of
positive operator-valued measurement (POVM), which is defined as
follows: The so-called POVM is a resolution of identity operator,
i.e. a collection $\mathbb{M}=\set{M_j}^N_{j=1}$ of nonnegative
operators that sum up to the identity operator, $\sum^N_{j=1} M_j =
\I$. We denote by $\cM$ all POVM on a quantum system. For a quantum
system prepared in a fixed state $\rho$, each POVM
$\mathbb{M}\in\cM$ performed on this system in the state $\rho$
induces a probability distribution $p^\mathbb{M}_\rho
=\Set{p^\mathbb{M}_\rho(j)}^N_{j=1}$, where $p^\mathbb{M}_\rho(j):=
\Tr{M_j\rho}$ is the probability of obtaining measurement outcome
$j$ and identified by Born's rule when a single measurement $M_j$ is
performed.

\section{Main result}

Our main result is the following:

\begin{thrm}\label{th:Fuchs-Zhang}
It holds that
\begin{eqnarray}\label{eq:lower-bound-1}
\rF(\rho,\sigma)\geqslant
1-\frac12\frac{e^{\frac12\rS_{\max}(\rho||\sigma)}}{1+e^{\frac12\rS_{\max}(\rho||\sigma)}}\norm{\rho-\sigma}_1
\end{eqnarray}
for $\rho,\sigma\in\density{\cH_d}$.
\end{thrm}

In order to prove Theorem~\ref{th:Fuchs-Zhang} we need a more
technical result, Theorem~\ref{th:zhang}, which may be of
independent interest.
\begin{thrm}\label{th:zhang}
Let $\rho$ and $\sigma$ be two density operators in
$\density{\cH_d}$, and $\lambda\in [0,1]$. Then
\begin{eqnarray}
\rF(\rho, \lambda\rho + (1-\lambda)\sigma) \geqslant
1-\frac12(1-\sqrt{\lambda})\norm{\rho-\sigma}_1.
\end{eqnarray}
\end{thrm}
The proof of Theorem~\ref{th:zhang} is based on the following
results:
\begin{prop}[\cite{Nielsen} equations (9.23) and (9.74)]\label{prop:optimum-of-measurement}
For given two states $\rho,\sigma\in \density{\cH_d}$, we have:
$$
\norm{\rho-\sigma}_1 =
\max_{\mathbb{M}\in\cM}\norm{p^\mathbb{M}_\rho -
p^\mathbb{M}_\sigma}_1 ~~~\text{and}~~~\rF(\rho,\sigma) =
\min_{\mathbb{M}\in\cM}\rF\Pa{p^\mathbb{M}_\rho,p^\mathbb{M}_\sigma}.
$$
\end{prop}

\begin{lem}\label{lem-1}
It holds that $\sqrt{(1+a)(1+\lambda a)}\geqslant 1+\sqrt{\lambda}a$
for $a,\lambda\geqslant0$.
\end{lem}

\begin{proof}
Now since $(1+a)(1+\lambda a) - (1+\sqrt{\lambda}a)^2 =
(1-\sqrt{\lambda})^2a\geqslant0$, the desired inequality follows
immediately.
\end{proof}

\begin{lem}\label{lem:zhang-bu}
Let $p=\set{p_j}^n_{j=1}$ and $q=\set{q_j}^n_{j=1}$ be two probability distribution, and $\lambda\in [0,1]$. Then
\begin{eqnarray}\label{eq:commuting-case}
\rF(p, \lambda p + (1-\lambda)q) \geqslant
1-\frac12(1-\sqrt{\lambda})\norm{p-q}_1.
\end{eqnarray}
\end{lem}

\begin{proof}
In fact, \eqref{eq:commuting-case} is equivalent to the following inequality:
\begin{eqnarray}\label{eq:prob-distribution}
\sum^n_{j=1}\sqrt{p_j(\lambda p_j + (1-\lambda)q_j)} + \frac{1-\sqrt{\lambda}}2\sum^n_{j=1}\abs{p_j-q_j}\geqslant1.
\end{eqnarray}
It suffices to show that \eqref{eq:prob-distribution} is true. Now
we introduce two sets as follows:
\begin{eqnarray}
\cI_>:=\set{j:p_j>q_j},~~~\cI_\leqslant:=\set{j:p_j\leqslant q_j}.
\end{eqnarray}
Thus we obtain
\begin{enumerate}[(i)]
\item $\sqrt{p_j(\lambda p_j + (1-\lambda)q_j)} \geqslant q_j + \sqrt{\lambda}(p_j-q_j)$ for all $j\in \cI_>$.
\item $\sqrt{p_j(\lambda p_j + (1-\lambda)q_j)} \geqslant p_j$ for all $j\in \cI_\leqslant$.
\end{enumerate}
Indeed, if $j\in \cI_>$, then $p_j>q_j$. Hence (i) apparently holds
when $q_j=0$. Without loss of generality, assume that $q_j\neq0$.
Setting $a = (p_j - q_j)/q_j$ in Lemma~\ref{lem-1} and multiplying
both sides with $q_j$ yields (i). The correctness of (ii) can be
easily seen.

One can see from the above facts that
\begin{eqnarray}
\sum^n_{j=1}\sqrt{p_j(\lambda p_j + (1-\lambda)q_j)}&\geqslant&\sum_{j\in \cI_>}\Br{q_j + \sqrt{\lambda}(p_j-q_j)} + \sum_{j\in \cI_\leqslant}p_j,\\
\frac{1-\sqrt{\lambda}}2\sum^n_{j=1}\abs{p_j-q_j}&=& \frac{1-\sqrt{\lambda}}2\sum_{j\in\cI_>}(p_j-q_j) + \frac{1-\sqrt{\lambda}}2\sum_{j\in\cI_\leqslant}(q_j-p_j)\\
&=&(1-\sqrt{\lambda})\sum_{j\in\cI_>}(p_j-q_j).
\end{eqnarray}
By easy computation, we can check the correctness of the following equality:
$$
\sum_{j\in \cI_>}\Br{q_j + \sqrt{\lambda}(p_j-q_j)} + \sum_{j\in
\cI_\leqslant}p_j + (1-\sqrt{\lambda})\sum_{j\in\cI_>}(p_j-q_j) = 1.
$$
This completes the proof.
\end{proof}

\begin{proof}[The proof of Theorem~\ref{th:zhang}]
In fact, for given two states $\rho$ and $\sigma$, each POVM
$\mathbb{M}=\set{M_j}^N_{j=1}$ in $\cM$ induces two probability
distributions $p^\mathbb{M}_\rho,p^\mathbb{M}_\sigma$. From
Lemma~\ref{lem:zhang-bu}, we know that
\begin{eqnarray}
\rF\Pa{p^\mathbb{M}_\rho, \lambda p^\mathbb{M}_\rho +
(1-\lambda)p^\mathbb{M}_\sigma} \geqslant
1-\frac12(1-\sqrt{\lambda})\norm{p^\mathbb{M}_\rho-p^\mathbb{M}_\sigma}_1.
\end{eqnarray}
Taking minimum over both sides of the above last inequality relative
to $\cM$, we obtain
\begin{eqnarray}
\min_{\mathbb{M}\in\cM}\rF\Pa{p^\mathbb{M}_\rho, \lambda
p^\mathbb{M}_\rho + (1-\lambda)p^\mathbb{M}_\sigma} &\geqslant&
\min_{\mathbb{M}\in\cM}\Pa{1-\frac12(1-\sqrt{\lambda})\norm{p^\mathbb{M}_\rho-p^\mathbb{M}_\sigma}_1}\\
&=&1-\frac12(1-\sqrt{\lambda})\max_{\mathbb{M}\in\cM}\norm{p^\mathbb{M}_\rho-p^\mathbb{M}_\sigma}_1,
\end{eqnarray}
implying the desired inequality by
Proposition~\ref{prop:optimum-of-measurement}.
\end{proof}

\section{The proof of main result}

The inequality of Theorem~\ref{th:zhang} is about the special pair
of states $\rho$ and $\lambda\rho+(1-\lambda)\sigma$ and seems to be
of rather restricted importance. However it is possible to
reformulate the inequality as an inequality about any pair of
states, which we will now show. Given two density operators $\rho$
and $\sigma$, we know that if the support of $\rho$ is contained in
the support of $\sigma$ i.e. $\supp(\rho)\subseteq \supp(\sigma)$,
then
$$
\min\set{\lambda>0: \rho\leqslant \lambda\sigma} =
\lambda_{\max}(\sigma^{-1/2}\rho\sigma^{-1/2}) := \lambda_0<\infty,
$$
we also know that $\min\set{\lambda>0: \rho\leqslant \lambda\sigma}
= +\infty$ if $\supp(\rho)\nsubseteq\supp(\sigma)$. Clearly
$\lambda_0>0$. If denote $\widehat\sigma:= \frac{\sigma -
\lambda^{-1}_0\rho}{1-\lambda^{-1}_0}$, then
$$
\sigma =  \lambda^{-1}_0\rho + (1- \lambda^{-1}_0)\widehat\sigma.
$$
Therefore
\begin{eqnarray*}
\rF(\rho,\sigma) = \rF(\rho,\lambda^{-1}_0\rho + (1-
\lambda^{-1}_0)\widehat\sigma) \geqslant
1-\frac12\Pa{1-\sqrt{\lambda^{-1}_0}}\norm{\rho - \widehat\sigma}_1
\end{eqnarray*}
implies that
$$
\rF(\rho,\sigma)\geqslant
1-\frac12\frac{\sqrt{\lambda_0}}{\sqrt{\lambda_0}+1}\norm{\rho-\sigma}_1~~\text{for}~~\lambda_0>0.
$$
The lower bound in the above inequality is indeed tighter than one
in Fuchs-van de Graaf's inequality. Thus, we get a state-dependent
factor in the lower bound for fidelity, that is, when $\lambda_0=
+\infty$, the above lower bound is reduced to the lower bound in
Fuchs-van de Graaf's inequality.

For the related problems along this line such as min- and max-
(relative) entropy, we refer to \cite{Datta}. By combing the
concavity of fidelity and Fuchs-van de Graaf's inequality, we have
$$
\rF(\rho,\lambda\rho+(1-\lambda)\sigma)\geqslant
1-\frac12(1-\lambda)\norm{\rho-\sigma}_1.
$$
Comparing this lower bound with ours indicates that our lower bound
is indeed tighter if $\lambda\in(0,1)$. That is
$$
\rF(\rho,\lambda\rho+(1-\lambda)\sigma)\geqslant
1-\frac12(1-\sqrt{\lambda})\norm{\rho-\sigma}_1\geqslant
1-\frac12(1-\lambda)\norm{\rho-\sigma}_1.
$$

\section{Discussion and conclusion}

In fact, we can also make analysis of the saturation of the first
inequality in \eqref{eq:Graaf} via our main result (i.e.
Theorem~\ref{th:Fuchs-Zhang}). Generally, we have
\begin{eqnarray}\label{eq:midle-bound}
\rF(\rho,\sigma)\geqslant
1-\frac12\frac{e^{\frac12\rS_{\max}(\rho||\sigma)}}{1+e^{\frac12\rS_{\max}(\rho||\sigma)}}\norm{\rho-\sigma}_1
\geqslant1-\frac12\norm{\rho-\sigma}_1.
\end{eqnarray}
Now we have equality if the lower bound on the fidelity between two
states $\rho$ and $\sigma$ is saturated in Fuchs-van de Graaf's
inequality. That is,
$\rF(\rho,\sigma)=1-\frac12\norm{\rho-\sigma}_1$, which means that
two inequality are saturated in \eqref{eq:midle-bound}. Thus we get
$$
\frac{e^{\frac12\rS_{\max}(\rho||\sigma)}}{1+e^{\frac12\rS_{\max}(\rho||\sigma)}}\norm{\rho-\sigma}_1
=\norm{\rho-\sigma}_1.
$$
This amounts to say that
$$
\frac{e^{\frac12\rS_{\max}(\rho||\sigma)}}{1+e^{\frac12\rS_{\max}(\rho||\sigma)}}=1~~\text{or}~~\norm{\rho-\sigma}_1=0,
$$
which is equivalent to $\rS_{\max}(\rho||\sigma)=+\infty$ or
$\rho=\sigma$. The condition $\rS_{\max}(\rho||\sigma) = +\infty$
actually says that the support of $\rho$ is not contained in the
support of $\sigma$. Although equality
$\rF(\rho,\sigma)=1-\frac12\norm{\rho-\sigma}_1$ implies
$\rS_{\max}(\rho||\sigma)=+\infty$ or $\rho=\sigma$, the converse is
not true \cite{Audenaert1}.

In this paper, we obtained a lower bound on the fidelity between a
fixed state and its a mixed path with another state (a similar topic
can be found in \cite{Audenaert2,Audenaert3}). Based on this result,
we derived a lower bound on the fidelity between two states, which
improved Fuchs-van de Graaf's inequality. Our main result answer
positively the conjecture proposed in \cite{Lin}. The potential
applications in quantum information theory are left in the future
research.

\subsubsection*{Acknowledgements}
The work was supported by National Natural Science Foundation of
China (11301124, 11171301) and by the Doctoral Programs Foundation
of the Ministry of Education of China (J20130061). The authors are
grateful to Koenraad M.R. Audenaert for sending us his note and
telling us his different approach towards the proof of
Theorem~\ref{th:zhang}.



\end{document}